\documentclass[sn-mathphys]{sn-jnl}

\usepackage{algpseudocode}
\usepackage{algorithm}
\usepackage{algorithmicx}
\usepackage{makecell}
\usepackage{graphicx}
\usepackage{subfigure}
\usepackage{lipsum}
\usepackage{amsmath,amssymb,amsthm}
\graphicspath{{fig/}}

\usepackage{url}

\newcommand\blfootnote[1]{%
  \begingroup
  \renewcommand\thefootnote{}\footnote{#1}%
  \addtocounter{footnote}{-1}%
  \endgroup
}

\jyear{2021}%

\theoremstyle{thmstyleone}%
\newtheorem{theorem}{Theorem}
\newtheorem{lemma}{Lemma}
\theoremstyle{thmstyletwo}%

\theoremstyle{thmstylethree}%
\newtheorem{myDef}{Definition}%

\begin{document}

\title[Fairness-driven Skilled Task Assignment with Extra Budget in Spatial Crowdsourcing]{Fairness-driven Skilled Task Assignment with Extra Budget in Spatial Crowdsourcing}


\author[1]{\fnm{Yunjun} \sur{Zhou}}
\author[1]{\fnm{Shuhan} \sur{Wan}}

\author*[1]{\fnm{Detian} \sur{Zhang}}

\author[2]{\fnm{Shiting} \sur{Wen}}

\affil*[1]{\orgdiv{Institute of Artificial Intelligence, School of Computer Science and Technology}, \orgname{Soochow University}, \orgaddress{\city{Suzhou}, \state{Jiangsu}, \country{China}}}

\affil[2]{\orgdiv{Ningbo Institute of Technology}, \orgname{Zhejiang University}, \orgaddress{\city{Ningbo}, \state{Zhejiang}, \country{China}}}



\abstract{With the prevalence of mobile devices and ubiquitous wireless networks, spatial crowdsourcing has attracted much attention from both academic and industry communities. On spatial crowdsourcing platforms, task requesters can publish spatial tasks and workers need to move to destinations to perform them. In this paper, we formally define the Skilled Task Assignment with Extra Budget (STAEB), which aims to maximize total platform revenue and achieve fairness for workers and task requesters. In the STAEB problem, the complex task needs more than one worker to satisfy its skill requirement and has the extra budget to subsidize extra travel cost of workers to attract more workers. We prove that the STAEB problem is NP-complete. Therefore, two approximation algorithms are proposed to solve it, including a greedy approach and a game-theoretic approach. Extensive experiments on both real and synthetic datasets demonstrate the efficiency and effectiveness of our proposed approaches. \blfootnote{*Corresponding author(s). E-mail(s): \url{detian@suda.edu.cn}.} \blfootnote{\quad Contributing authors: \url{20214227073@stu.suda.edu.cn};} \blfootnote{\quad \url{20195227042@stu.suda.edu.cn}; \url{wensht@nit.zju.edu.cn};}}




\maketitle

\section{Introduction}
	\label{Sect.1}
	With the development of smart devices and high-speed wireless networks, spatial Crowdsourcing (SC) assigning moving workers to location-based tasks has recently gained much attention. Specifically, workers need to physically move to the specified locations to accomplish the task which published by task requesters. Such spatial crowdsourcing can be used in many applications, such as online taxi-calling services (e.g., DiDi and Uber), food delivery services (e.g., Grubhub, Eleme and Meituan), traffic monitoring (e.g., Waze) and geographical data generation (e.g., OpenStreetMap). However, some complex tasks require not only specific skills but also the need to ensure that workers are within a certain distance to ensure a fair gain.
	
	A fundamental problem in spatial crowdsourcing is task assignment. Most of the existing works on task assignment \cite{zhao2020predictive,to2016real,chen2020fair,liu2020budget,tong2019two,zheng2020Online} mainly assume that all the tasks are simple, and can be easily completed by a single worker such as delivering packages, taking photos, or reporting hot spots. However, an individual worker cannot complete some complex tasks, e.g., preparing for a party, decorating houses, and general cleaning. For example, decorating a house requires design drawings, painting walls, tiling, moving furniture, installing water and electricity, et al. Therefore, the platform needs to arrange multiple workers with different skills to meet the different needs of tasks. Furthermore, it is important to take into account the fairness of the assignment. Because for the same task, each worker not only needs to provide skill service but also has to pay different travel costs.
	
	Previous works~\cite{cheng2016taskmulti,cheng2019cooperation} on spatial crowdsourcing have focused on assigning complex tasks to multiple workers such that these workers can cover the skill needs of tasks while ignoring the additional travel cost of the workers. However, in reality, if workers want to complete the assigned task, they also need to move to the location of the task, which will incur travel costs. Then workers are only willing to accept the travel cost within a certain budget so that their actual earnings are fair. Take Fig. \ref{fig} for example, there are three workers with different skills (denoted by different colors), and two task requesters that require appropriate skills to complete (the required skills are also denoted by the different colors). The fixed range constraints of the task are shown with solid lines, and the extra range constraints of the task are shown with dashed lines. We can see, worker $w_2$ can well meet the skill requirement of task requester $t_2$ and does not exceed the distance limit. However, if task requester $t_1$ wants to find a worker to complete his task, no worker can answer it, because the worker (i.e., $w_1$) who meets the skill requirement exceeds his distance limit, while the skill of the worker (i.e., $w_3$) who within the distance limit do not meet the requirements. It results that task requester $t_1$ may never be assigned. In practice, task requester $t_1$ may have an extra budget to subsidize the extra travel cost for worker $w_1$ so that the task can be completed, which can also facilitate a fair online recommendation.

    \begin{figure}[h]
    \centering
    \includegraphics[width=0.5\textwidth]{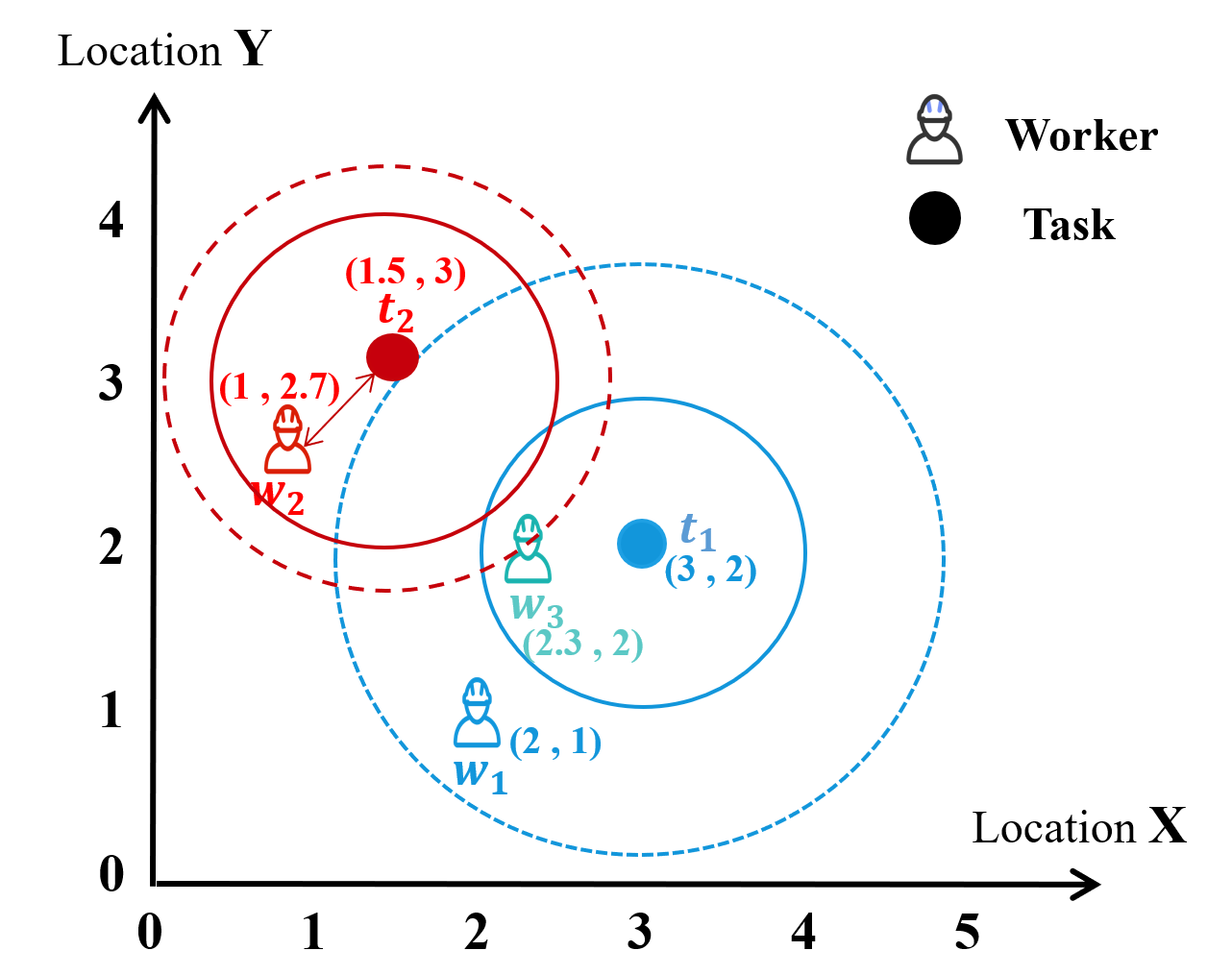}
    \caption{Example of distance constraint in STAEB problem.} \label{fig}
    \end{figure}

	Therefore in this paper, we investigate the task assignment of SC under such a problem setting, namely Skilled Task Assignment with Extra Budget (STAEB). To be more specific, given a set of workers and a set of tasks, it aims to assign multi-skilled workers to complex tasks with extra budget.

	To summarize, we make the following contributions to this paper:
	\begin{itemize}
		\item
		We formally define the Skilled Task Assignment with Extra Budget (STAEB) problem and prove it is NP-complete.
		\item
		We propose two batch-based approximation algorithms to solve the STAEB problem, i.e.,
		greedy and game-theoretic approaches.
		\item
		We conduct extensive experiments on real and synthetic datasets to prove the effectiveness and efficiency of our algorithms.
	\end{itemize}
	
	The rest of this paper is organized as follows. Section \ref{Sect.2} reviews some related work. The STAEB problem is formally defined in Section \ref{Sect.3}. Section \ref{Sect.4} gives the greedy algorithm. The game-theoretic algorithm is proposed in section \ref{Sect.5}. Extensive experiments on real and synthetic datasets are presented in section \ref{Sect.6}.  Finally, section \ref{Sect.7} concludes this work.

\section{Related Work}
\label{Sect.2}
Spatial crowdsourcing is an important topic in match-based services, which has attracted much attention from scholars in the mobile Internet and sharing economy area. Among them, task assignment is one of the most important areas in spatial crowdsourcing and can be classified into two categories~\cite{tong2020spatial}: Matching and Planning. In the matching model, task assignment is often formulated as a bipartite graph-based problem. Workers and tasks can be represented by the vertices in the bipartite graph, and utility or cost between a worker and a task can be denoted by the weight of the edges. Then, the problem aims to obtain an optimal matching in the bipartite graph. In the planning model (a.k.a. scheduling model), task assignment aims to plan a route for each worker to perform a sequence of tasks. In this section, we introduce the related work about spatial crowdsourcing based on the above two categories.\par
\subsection{Matching Model}
Usually, the matching problem can be categorized as utility maximization \cite{tong2020spatial, kazemi2012geocrowd,waissi1994network}, cost minimization \cite{burkard2012assignment,derigs1981shortest,leong2008capacity}, and stable matching \cite{corral2000closest,wong2007efficient} according to different objectives in existing studies.\par

The objective of utility maximization is equivalent to either maximizing the total number of assignments or the total payoff \cite{tong2020spatial}. Since task matching can be formulated as a bipartite graph matching problem, so the exact algorithm in the bipartite matching problem (e.g., Hungarian algorithm \cite{burkard2012assignment}) can optimally solve this problem. Kazemi et al. \cite{kazemi2012geocrowd} obtain the exact result by reducing the graph into an instance of the maximum flow problem \cite{waissi1994network}, and using the Ford-Fulkerson algorithm \cite{ford1956maximal}. Besides, various greedy-based algorithms are proposed to reduce the computation of the exact algorithm (e.g., Hungarian \cite{burkard2012assignment}). To et al. \cite{to2016real} consider greedy matching based on the priority of each task in this problem. Specifically, the priority is to expand by the idea of location entropy \cite{kazemi2012geocrowd} to region entropy, i.e., tasks with fewer workers inside should have a higher priority to be assigned. \par

Since each worker-task assignment has a budget, finding a match with minimum cost is also important in task assignment. Hungarian algorithm \cite{burkard2012assignment} and successive shortest path algorithm (SSPA) \cite{derigs1981shortest} can get the exact result of this problem. Besides, Hou et al. \cite{leong2008capacity} leverage indexing (R-tree indexing) and I/O optimization techniques to improve efficiency.\par 
Stable matching tries to integrate the preferences of either workers or tasks into their optimization objectives. The intuitive solution is to iteratively select the closest pair from the remaining tasks and workers \cite{corral2000closest}.  To improve efficiency, Wong et al. \cite{wong2007efficient} reduce the concept of “mutual nearest neighbor” to the bichromatic mutual NN search problem, and propose an NN search-based chain algorithm.\par

As discussed above, the matching model is more like a bipartite graph-based problem, which only assigns one worker for each task. However, in our problem, a task may need multiple different skilled workers to be accomplished, so those studies are not suitable for our problem.

\subsection{Planning Model}
The planning model in task assignment is to plan a route (i.e., a sequence of tasks) for each worker. Food delivery and ride-sharing are both planning problems in real applications. We can further divide the model into two categories according to the task complexity (i.e., whether the task requires the cooperation of multiple workers) in the planning model: single-worker planning \cite{zhao2020predictive,to2016real,chen2020fair,liu2020budget} and multi-workers planning \cite{he2014toward,to2015server,cheng2016taskmulti}.
\subsubsection{Single-worker planning}
Most of the previous works mainly focus on assigning the task to a single worker \cite{zhao2020predictive,to2016real,chen2020fair,liu2020budget,tong2019two,zheng2020Online}. Kazemi et al.~\cite{kazemi2012geocrowd} reduce the matching to the maximum flow problem~\cite{ahuja1988network}, and use Ford-Fulkerson algorithm~\cite{ford1956maximal} to maximize the number of assigned tasks. Den et al.~\cite{deng2013maximizing} firstly study maximizing the number of performed tasks under travel budget and deadline constraints, also they prove its NP-hardness. To address the problem, they propose an exact algorithm based on dynamic programming and several approximation algorithms based on the greedy heuristic. The greedy-based heuristic algorithms include the nearest-neighbor heuristic (NNH), most promising heuristic (MPH), and least expiration time heuristic (LEH). To better achieve the trade-off between efficiency and effectiveness, the beam search heuristic (BSH) is proposed~\cite{deng2016task}, which extends the base of the candidate set to a given threshold in the NNH. Tong et al.~\cite{tong2016online} first consider the online scenario of task assignment and propose threshold-based algorithms with theoretical guarantees to maximize the total utility of the assignment. Song et al.~\cite{song2017trichromatic} present trichromatic online matching in real-time spatial crowdsourcing which contains three entities of workers, tasks, and workplaces. Cheng et al.~\cite{cheng2020real} propose a cross-online matching that enables the platform to borrow some unoccupied workers from other platforms. Zhao et al.~\cite{peng2020user} first consider reducing the average waiting time of users, and many complete tasks so as to improve the user experience.\par
\subsubsection{Multi-workers planning}
In practice, there are some complex tasks that require groups of workers to conduct, for one worker usually a worker does not have all the skills. Then, some researchers focus on multi-worker planning (i.e. multiple workers collaborate on a task). The planning modes are to maximize general utility (e.g., satisfaction scores~\cite{gao2017team,she2015utility}, payoffs~\cite{he2014toward,asghari2016price,tao2018multi,zheng2018order}, distance~\cite{huang2013large,ma2013t,asghari2017line}).\par
In detail, Shin et al.~\cite{he2014toward} propose a local ratio-based algorithm, to maximize the reward of the performed tasks. To et al.~\cite{to2015server} formally define the maximum task assignment (MTA) problem in spatial crowdsourcing and propose alternative solutions to address it. Then, considering that different types of tasks may require workers with different skill sets, they first extend the MTA by the skills of the workers, which namely the maximum score assignment (MSA) problem. Cheng et al.~\cite{cheng2016taskmulti} propose the multi-skill spatial crowdsourcing (MS-SC) problem, which finds an optimal worker-and-task assignment strategy. MS-SC is proven NP-hard. Therefore, they propose three effective heuristic approaches, including greedy, g-divide-and-conquer, and cost-model-based adaptive algorithms to solve it. She et al.~\cite{she2015utility} also propose a greedy algorithm based on RatioGreedy, which considers the utility-cost ratio of each worker-task pair and adds the pair with the largest ratio to the planning. Gao et al.~\cite{gao2017team} consider both skill and time constraints in order to maximize total satisfaction. They first form a set of workers with the lowest base to meet the skill requirements of the task, and then greedily assign the highest satisfaction workers to the task. Gao et al.~\cite{gao2017Top} recommend top-k groups of workers to the task and choose one worker in each group to lead the group. Cheng et al.~\cite{cheng2019cooperation} consider the cooperative relationship of multiple workers, to maximize the total cooperative quality income of tasks. The online algorithm for real-time spatial crowdsourcing is first developed by Song et al.~\cite{song2019Multi}. The Online-Exact algorithm always computes the optimal assignment for the newly appearing tasks or workers and the Online-Greedy algorithm is for fast computing task assignment. Li et al.~\cite{Li2020Group} propose group task assignments considering group preference for tasks. Ni et al.~\cite{ni2020Task} consider the tasks may have some dependencies among them. That is one task can only be dispatched when its dependent tasks have been assigned.\par 
However, these methods do not consider the urgent need for tasks with extra budget (as depicted in Fig. 1), therefore some tasks may need to wait for a long time and even may never get served. This is actually unfair for these tasks and causes a bad experience for platform customers. In conclusion, all of these existing papers mainly focus on task-worker matching, while ignoring the fairness issues that exist (e.g., workers far from the task). In this paper, we take the extra budget of tasks into consideration and propose fairness-driven task assignment algorithms to efficiently and effectively solve the problems.

\section{Problem Definitions}
\label{Sect.3}
	In this section, we formally define the Skilled Task Assignment with Extra Budget (STAEB) problem.
	
	Assume that $S=\{s_1,s_2...s_k\}$ is a set of $k$ skills. Each worker has one or multiple skills in $S$ and each task needs one or multiple skills in $S$. Different skills require task requesters to pay different fees $p_s$.
	
	\begin{myDef}
		\label{def1}
		(Task) A task, denoted by $t=<l_t,a_t,r_t,b_t,S_t>$, is released on the platform with location $l_t$ in the 2D space at time $a_t$. $r_t$ is the radius of $t$ which is the fixed range constraint of $t$, $b_t$ is its provided extra budget and $t$ needs skills $S_t \subseteq S$.
	\end{myDef}

	\begin{myDef}
		\label{def2}
		(Worker) A worker, denoted by $w=<l_w,a_w,S_w>$, appears on the platform at time $a_w$ and at location $l_w$ in the 2D space. $w$ has skills $S_w \subseteq S$.
	\end{myDef}

	\begin{myDef}
		\label{def3}
		(Travel cost) The travel cost, denoted by $cost(t,w)$, is determined by the travel distance from $l_w$ to $l_t$.
	\end{myDef}
	
	Travel distance can be measured by any type of distance such as Euclidean distance or road network distance. In this paper, we use Euclidean distance as the travel distance and take it as the travel cost directly for simplicity.
	\begin{myDef}
		\label{def4}
		(Extra travel cost) Extra travel cost, denoted by $e_t$, is the actual travel cost exceeding the fixed range constraint, i.e., $e_t = cost(t,w) - r_t$.
	\end{myDef}
	
	\begin{myDef}
		\label{def5}
		(Valid worker set) Each task needs to be completed by multiple workers. The set of multiple workers who can complete the task is called a valid worker set and the skill set of workers in a valid worker set can cover the skills that the task requires. Each skill requirement of the task only needs to be matched by one worker, but skills may be duplicated between workers. Therefore, each worker in a valid worker set must have at least one skill that the task needs but other workers do not have. The valid worker set, denoted by $W_{v}(t)$, is a set of workers that satisfy the following three conditions:
		
		(1)Each skill of the task can be covered by the skills in the set, i.e., $s \in \bigcup_{w\in W_{v}(t)}S_{w}, \forall s \in S_t$,
		
		(2)The total extra travel cost of workers in the set should be less than the extra budget of the task, i.e., $\sum_{w \in W_{v}(t)} e(t,w) \leq b_t$.
	\end{myDef}

	\begin{myDef}
		\label{def6}
		(Platform revenue) Given a task $t$ and a worker $w$, the platform revenue is denoted as:
		\begin{equation}
			p_{(t,w)}=\left\{
			\begin{aligned}
				\alpha \sum_{s \in S_w \cap S_t} p_s & , & cost(t,w)<r_t, \\
				\alpha \sum_{s \in S_w \cap S_t} p_s - \beta e(t,w) & , & cost(t,w) \ge r_t.
			\end{aligned}
			\right.
		\end{equation}
	\end{myDef}
where $\sum_ {s \in S_w \cap S_t} p_s$ represents the total revenue of providing skills the task needs by the worker. $\alpha (0 < \alpha < 1) $ is the parameter that controls the platform's income from the skill fee. We assume that the platform skill income is proportional to the skill fee. Workers are only willing to accept a certain range of travel cost. Therefore, the task and the platform will subsidize workers' extra travel cost. $\beta (0<\beta<1)$ is the parameter that controls the subsidy of the platform for the workers' extra travel cost.

When the travel cost of the worker to the task is less than the fixed range constraint of the task, the platform income is equal to the bonus of skill revenue. When the travel cost of the worker to the task is greater than or equal to the fixed range constraint of the task, the platform needs to subsidize the worker for the travel cost over the fixed range constraint of the task. Therefore, the platform revenue is equal to skill revenue minus extra travel costs. The platform revenue from the task is equal to the revenue of the workers assigned to the task, denoted as:

\begin{equation}
	P_{<t,W_{v}(t)>}=\sum_{w \in W_{v}(t)} p_{(t,w)}
\end{equation}

\begin{myDef}
	\label{def7}
	(Skilled Task Assignment with Extra Budget (STAEB) problem) Given a set of tasks $T$ and set workers $W$. A feasible matching result, denoted by $M$, consists of a set of $<$ \textit{task}, \textit{valid worker set} $>$ in the form of $<t_1,W_{v}(t_1)>, <t_2,W_{v}(t_2)>,...,<t_{\lvert T \rvert},W_{v}(t_{\lvert T \rvert})>$, where $\cap _{i=1}^{\lvert T \rvert} W_{v}(t_i) = \emptyset$. Our problem is to find a feasible matching result $M$ that achieves the goal of maximizing total platform revenue fairly $P_M=\sum_{<t, W_{v}(t)> \in M} P_{<t, W_{v}(t)>}$.
		
\end{myDef}

\begin{theorem}
	The Skilled Task Assignment with Extra Budget problem is an NP-complete problem.
\end{theorem}

\begin{proof}
	We prove that the Skilled Task Assignment with Extra Budget problem is NP-complete by reducing it to a maximum weighted independent set problem. In our problem, we can first initialize each available $W_v(t)$ of all $t \in T$ to be a vertex in the graph. Then let all $W_v(t)$ vertices belonging to the same $t$ be connected one by one. Finally, we connect the $W_v(t)$ vertices that have intersections (i.e., have the same workers). At this point, any set of $W_v(t)$  vertices that can be assigned to T must be an independent set (otherwise there would be at least one t that has more than one $W_v(t)$, or two $W_v(t)$ with the same worker). Since there is revenue for each $W_v(t)$, the revenue of the platform is also the sum of the vertex weights on the independent subset, the STAEB problem is reduced successfully to the maximum-weighted independent set problem. Because the maximum weight-independent subset problem is an NP-complete problem, the STAEB problem is also an NP-complete problem.
\end{proof}

\section{Greedy algorithm}
\label{Sect.4}
In this section, we propose a greedy method \cite{ni2020Task}. The main idea of the algorithm is to sort tasks in descending order of average fee of skills and then find the fewest workers to cover the skills required by the task.	
The fees for different skills are different, so workers will get higher revenue if they complete skills with a higher fee. If a worker can meet multiple skills needed for a task, it will save travel costs. Therefore, for each task, the smallest set of valid workers is preferred.

	\begin{algorithm}[htbp]
		\caption{Greedy algorithm} 
	      \hspace*{0.02in} {\bf Input:} 
		A set of tasks $T$, a set of workers $W$\\
		\hspace*{0.02in} {\bf Output:} 
		the matched pair set $M$
		\begin{algorithmic}[1]
			\State $M\leftarrow \emptyset$, $W^{'} \leftarrow W$;
			\State Sort the tasks in descending order according to the size of $\frac{\sum_{s \in S_t} p_s}{\lvert number\ of\ task\ skills\rvert}$; 
			\For{each task $t\in T$}
				\State $W_{v}(t) \leftarrow \emptyset$;
				\While{$S_t \neq \emptyset$}
					\State $R_t=r_t+b_t$;
					\State $W^*= \{ w \vert S_w \cap S_t \neq \emptyset and cost(t,w)\leq R_t \}$;
					\If{$W^{*} == \emptyset$}
						\State $W \leftarrow W+W_{v}(t)$;
						\State $W_{v}(t) \leftarrow \emptyset$;
						\State Break;
					\EndIf
					\State $w= argmax_{w \in W^*}\lvert S_w \cap S_t \rvert $;
					\State $W_{v}(t)=W_{v}(t) \cup \{w\}$;
					\State $b_t=b_t-e(t,w)$;
					\State $S_t=S_t-\{S_w \cap S_t\}$;
					\State $W=W-\{w\}$;				
				\EndWhile
				\If{$W_{v}(t) \neq \emptyset$}
					\State $M \leftarrow M \cup \{<t,W_{v}(t)>\}$;
				\EndIf		
			\EndFor
			\State \Return $M$;
		\end{algorithmic}
	\end{algorithm}

Algorithm 1 presents the detailed steps of the Greedy algorithm. In line 1, the matched pair set is initialized to empty and the initial unmatched worker set is set to the inputs of the worker. In line 2, the tasks are sorted in descending order according to the average fee of skills. The valid worker set of each task is set to empty in lines 3-4. When the unmatched skills of tasks are not empty, we will assign workers as follows: in lines 6-7, the total range constraint of the task is a fixed range constraint plus an extra range constraint. $W^*$ is the set of workers who can complete the task. If the set of workers who can complete the task is empty, it means that no worker can complete the remaining skills of the task, so the task cannot be completed and the workers in the valid worker set $W_v(t)$ are added to the worker set $W^{'}$. Then in lines 13-17, we select the worker $w$ in $W^{*}$ who has the most skills required by the task and add $w$ to the valid worker set of the task. We update the extra range constraint and the skill set required by the task and remove the assigned workers from the worker set.  In lines 19-20, when the valid worker set is not empty and the skill set is empty, it indicates that the valid worker set can cover the skills of the task. We can add this task and its valid worker set to the matched pair set and remove the task from the task set. Finally, we return the matched pair set in line 23.

\textbf{Complexity Analysis.} The time complexity of sorting tasks is $O(\vert S_t \vert \cdot \vert T\vert + \vert T \vert \cdot \ln \vert T\vert)$ and finding valid worker sets for each task is $O(\vert T \vert \cdot \vert S_t \vert^2 \cdot \vert W \vert)$, so the total time complexity is $O(\vert S_t\vert\cdot \vert T\vert + \vert T\vert\cdot\ln \vert T\vert+\vert T\vert\cdot \vert S_t\vert^2 \cdot \vert W\vert)$.

\section{Game-Theoretic Algorithms}
\label{Sect.5}
The greedy algorithm can find the results effectively. However, it assumes that workers will complete the tasks according to the assignment of the platform, and the competition between different workers is ignored. The fundamental nature of the STAEB problem is that each task needs to be completed by multiple workers who satisfy the need for skills, which means that the choice of workers is affected by the decisions made by other workers. Each worker wants to choose a task with high revenue. This interdependent decision can be modeled through game theory, in which workers can be regarded as independent players participating in the game. Thus, the STAEB problem can be formalized as a multi-player game. There have been many game theory model studies~\cite{Fudenberg1992Game,Myerson1997Game,Rasmusen2006Games}. Based on the existing studies in game theoretic models, we propose a game theory method to find the valid worker sets to the task until the Nash equilibrium is satisfied in this section. 

\subsection{Game Formulation}
Our STAEB problem can be formulated as an $n-$player strategic
game $\mathcal{G}=<W,\mathbb{Y},\mathbb{U}>$, which consists of players, the overall strategies, and utility functions. It is formulated as follows:

(1) $W=\{w_1,w_2,...,w_n\}(n \ge 2)$ is a limited set of workers as game players. In the rest of
the paper, we will use player and worker interchangeably.

(2) $\mathbb{Y}=\cup _{i=1}^{n} Y_i$ is the overall strategies for the players. $Y_i$ is the finite
set of strategies that worker $w_i$ can choose.

(3) $\mathbb{U}=\cup _{i=1}^{n} U_i$ denotes the utility functions of all the players. The value of $U_i$ depends on the player $w_i$ and other players' strategies. $U_i:\mathbb{Y} \rightarrow \mathbb{P}$ is the utility function of player $w_i$. For every joint strategy $\vec{y} \in \mathbb{Y}$, $U_i(\vec{y}) \in \mathbb{P}$ represents the utility of player $w_i$, which can be calculated as follows:

\begin{equation}
	U_{i}(\overrightarrow{y})=p_{(t,w_i)}-p_{(t_0,w_i)}
\end{equation}
where $p_{(t,w_i)}$ is the platform revenue of assigning player $w_i$ to task $t$. $p_{(t_0,w_i)}$ is the as platform revenue of assigning player $w_i$ to task $t_0$.

Let $y_i$ be the strategy of player $w_i$ in the joint strategy $Y_i$ and $y_{-i}$ be all other players’ joint strategies except for player $w_i$. A strategic game has a pure Nash equilibrium $Y^* \in \mathbb{Y}$ if and only if for every player $w_i \in W$ satisfies the following conditions:
\begin{equation}
	U_i({y_i}^*,{y_{-i}}^*) \ge U_i(y_i,{y_{-i}}^*),\forall y_i \in Y_i, y_{i}^{*} \in Y^{*}
\end{equation}

In our problem, since each worker needs to have a deterministic strategy, i.e., selecting a task or doing nothing. We only consider deterministic strategies, which means that the probability of
a multiplayer worker $w_i$ can choose from $Y_i$ is 1, while the probabilities of the remaining strategies in $Y_i$ are 0. Then, we prove that the STAEB game is an Exact Potential Game (EPG) \cite{1996Potential} which has at least one pure Nash Equilibrium.

In a Nash equilibrium, no player can improve his utility by unilaterally changing his strategy when other players insist on their current strategies. We first introduce the theory of the Exact Potential Game.
\begin{myDef}[Exact Potential Game]
	\label{def8}
	A strategic game $ \mathcal{G}=<W,\mathbb{Y},\mathbb{U}>$ is called an exact potential game if and only if there exists a
	potential function $\phi:\mathbb{Y} \rightarrow \mathbb{P}$, such that for all $\vec{y} \in \mathbb{Y}$ and $ w_{i} \in W $:
	\begin{equation}
		U_i(y_i,y_{-i}) - U_i({y_i}',y_{-i})=\phi(y_i,y_{-i}) - \phi({y_i}',y_{-i}), \forall y_i,{y_i}' \in Y_i
	\end{equation}
\end{myDef}
where $y_i$ and ${y_i}'$ are the strategies that worker $w_i$ can choose, $y_{-i}$ is the joint strategy of other workers except for worker $w_i$. 
\begin{theorem}
	Our STAEB problem is an Exact Potential Game(EPG).
\end{theorem}
\begin{proof}
The total platform revenue of all tasks in $T$ is represented by the potential function $\phi(y) = \sum_{t \in T}{P_{<t, W_{v}(t)>}}$.  Worker $w_i$ can choose between $y_i$ and ${y_i}'$ strategies, while $y_{-i}$ is the strategy of all other workers except worker $w_i$. The task chosen in strategies $y_i$' is denoted by $t_j$ and $t_k$. Then we have:
	\begin{equation}
		\begin{aligned}
			&\phi(y_i,y_{-i}) - \phi({y_i}',y_{-i})\\	
			&=p_{<t_j,W_{v}(t_j)>}+p_{<t_k,W_{v}(t_k)-w_i>}+\sum_{t \in T-t,_j-t_k} p_{<t,W_{v}(t)>}-(p_{<t_k,W_{v}(t_k)>}\\ 
			&+p_{<t_j,W_{v}(t_j)-w_i>}+\sum_{t \in T-t_j-t_k}p_{<t,W_{v}(t)>})\\
			&=p_{<t_j,W_{v}(t_j)>}+p_{<t_k,W_{v}(t_k)-w_i>}-(p_{<t_k,W_{v}(t_k)>}+p_{<t_j,W_{v}(t_j)-w_i>})\\					&=p_{<t_j,W_{v}(t_j)>}-p_{<t_j,W_{v}(t_j)-w_i>}-(p_{<t_k,W_{v}(t_k)>}-p_{<t_k,W_{v}(t_k)-w_i>})\\
			&=p_{(t_j,w_i)}-p_{(t_k,w_i)}\\	
			&=(p_{(t_j,w_i)}-p_{(t_0,w_i)})-(p_{(t_k,w_i)}-p_{(t_0,w_i)})\\
			&=U_i(y_i,y_{-i}) - U_i({y_i}',y_{-i})
		\end{aligned}	
	\end{equation}
	Thus, the strategic game of the STAEB problem is an exact potential game, according to the definition.
\end{proof}

\subsection{The Game Theoretic Approach}
Because our STAEB game has pure Nash equilibrium, we propose an Extra Budget-aware Game-Theoretic method (EBGT) based on the best response framework to find the Nash equilibrium joint strategy of the strategic game $\mathcal{G}$. In this algorithm, each worker is assigned to his/her "best" task, so as to obtain a higher total platform revenue. The details of game theory method are described as two steps in algorithm 2. 

\begin{algorithm}[htbp]
	\caption{Extra Budget-aware Game-Theoretic (EBGT) Approach} 
	\hspace*{0.02in} {\bf Input:} 
	A set of tasks $T$, a set of workers $W$\\
	\hspace*{0.02in} {\bf Output:} 
	the matched pair set $M$
	\begin{algorithmic}[1]
            \State \textbf{// Step 1: Initialize the strategy.} 
		\State $M \leftarrow \emptyset$, $T' \leftarrow T$, $W' \leftarrow W$;	
		\For{each task $t\in T$}
			\State $W_{v}(t) \leftarrow \emptyset$;
			\For{each task skill $s\in S_t$}
				\State $R_t=b_t$;
				\State $W^*=\{w \vert s\in S_w$ and $cost(t,w)\leq R_t\}$;
				\If{$W^*= \emptyset$}
					\State $W' \leftarrow W'+W_{v}(t)$;
					\State $W_{v}(t) \leftarrow \emptyset$;
					\State Break;
				
				\EndIf
				\State $w=argmax_{w \in W^*} {P_{(t,w)}}$;
				\State $W_{v}(t)=W_{v}(t) \cup \{w\}$, $b_t=b_t-e(t,w)$\, $W'=W'-\{w\}$;
			
			\EndFor
			\If{$W_{v}(t) \neq \emptyset$}
				\State $M \leftarrow M \cup \{<t,W_{v}(t)>\}$, $T'=T'-\{t\}$;
			\EndIf
		
		\EndFor
            \State \textbf{// Step 2: Find the Nash equilibrium.} 
		\State \textbf{repeat}
		\State $M' \leftarrow M$;
		\For{each worker $w\in W$}
			\State find the best-response task $t^*$ for $w_i$; 
			\If{$t^*$ not exists}
				\State Continue;			
			\Else			
				\If{$t^* \in T'$}
					\State obtain $W_{v}(t^*)$ containing $w_i$;
					\State $M \leftarrow M \cup \{<t^*,W_{v}(t^*)>\}$, $T'=T'-\{t^*\}$;			
				\Else				
					\State $W'\leftarrow W'+W_{v}(t^*)$;
					\State $M \leftarrow M - \{<t^*,W_{v}(t^*)>\}$;
					\State $W_{v}(t^*) \leftarrow \emptyset$;
					\State obtain $W_{v}(t^*)$ containing $w_i$;
					\State $M \leftarrow M \cup \{<t^*,W_{v}(t^*)>\}$, $T'=T'-\{t^*\}$;
				\EndIf		
			\EndIf 
			\State compare $M$ and $M'$;	
		\EndFor
		\State \textbf{until} Nash equilibrium
		\State Update $M$;
		\State \Return $M$;		
	\end{algorithmic}
\end{algorithm}



\textit{Step 1: Initialize the strategy.} In line 2, the matching set is initialized to empty, and the initial unmatched task and worker set are set to the inputs of task and worker, respectively. In lines 3-4, the valid worker set of each task is set to empty. We traverse each skill of the task and find the workers who can complete the task in lines 5-12, if the set of workers who can complete the skill task is empty, it means that there is no worker who can complete the remaining skills of the task. Therefore, the task cannot be completed, and we restore the unmatched worker set. In lines 13-14, we select the workers with the highest platform revenue in $W^{*}$, and update the skill set required by the task and its extra range constraint. Meanwhile, we remove the assigned workers from the worker set. When the valid worker set can cover the skills of the task, we add the task and its valid worker set to the matched pair set in line 17.

\textit{Step 2: Find the Nash equilibrium.} The algorithm adjusts each worker’s strategy iteratively to get the best response strategy based on the current joint strategy of other workers until the Nash equilibrium is found in which no one will change his strategy. By this time, the platform revenue will be maximized. In each iteration, only one worker is allowed to choose the best game, and the game should be carried out in order. We record the matching results of the previous round in line 22. For each worker $w_i \in W$, we first find the best response task $t^{*}$ in line 24, which can be calculated as follows:
\begin{equation}
	\begin{aligned}
		&t^*=argmax_{t \in \{t \vert S_{w_i} \cap S_t \neq \emptyset\ and\ cost(t,w_i)  \le R_t\}} (p_{<t,W_{v}(t)>}-p_{<t,W_{v}(t)-\{w_i\}>}\\
		&-(p_{<t_0,W_{v}(t_0)>}-p_{<t_0,W_{v}(t_0)-\{w_i\}>}))
	\end{aligned}
\end{equation}

When there is no best response task for the worker based on the current task assignment, $w_i$ does not change his strategy. In lines 28-30, when $w$ has the best response task and the valid worker set of the best response task is empty, we obtain the valid worker set which includes $w_i$. We add $t^{*}$ and a valid worker set to the matched set. In lines 31-38, when the valid worker set of the best response
task is not empty, we clear the current valid worker set of the task and obtain the valid worker set which includes $w_i$. Then $t^{*}$ and a valid worker set to the matched set are added. Finally, we update M according to the Nash equilibrium.

\subsection{Analysis of the Game Theoretic Approach}
Since the STAEB problem is an exact potential game and the strategy set $S$ is limited, a Nash equilibrium can be reached after workers change strategies a limited number of rounds. For simplicity, we prove the upper bound of the total rounds required to achieve a pure Nash equilibrium by considering a scaled version of the problem. 
To verify the upper bound of the total rounds required to achieve a pure Nash equilibrium, we explore a scaled version of the issue. The objective function of the problem is an integer value. 
We suppose that a comparable game with potential function exists: $\phi_{\mathbb{Z}} (S) =d \cdot \phi (S)$, in which $d$
is a positive multiplicative factor such that $\phi_{\mathbb{Z}} (S) \in \mathbb{Z}$, $\forall S \in \mathbb{S}$. We show that the GT technique performs at most $\phi {\mathbb{Z}} (S^*)$ rounds using this scaled potential function, where $S^*$ is the optimal strategy that workers can select in this potential STAEB game. $\phi_{\mathbb{Z}} (S^*)$ is the product of the positive multiplier factor $d$ and the optimal value of the objective function $P_M$.

\begin{lemma} 
\label{lemma1}
The upper bound on the number of rounds to converge to a pure Nash equilibrium is $\phi_{\mathbb{Z}}(S^*)$  in each batch with the EBGT technique, where $S^*$ is the optimal joint strategy that workers can select in the potential STAEB game, and $\phi {\mathbb{Z}} (S^*) =d \cdot \phi (S^*)$ is a scaled potential function with only integer values.
\end{lemma}

\begin{proof}
	The EBGT approach converges when no worker deviates from his current strategy, which means that there is at least one worker deviates from his current strategy in each round. Because $\phi_{\mathbb{Z}} (S) \in \mathbb{Z}$, each worker $w_i$ is changed from its current strategy ${s_i}'$ to a better strategy $s_i$, which will increase the scaled potential function by at least 1, i.e., $U_i (s_i,s_{-i})-U_i ({s_i}',s_{-i}) \geq 1$. Thus the upper bound of the number of rounds to converge to pure Nash equilibrium is the maximum value $\phi_{\mathbb{Z}} (S^*)$.
	
	The upper bound of the platform revenue of the task $t_j$ in our best joint strategy is:
	\begin{equation}
		\overline p_{t_j} = \sum_{w_i \in W_{v}(t_j)} p_{(t_j,w_i)}
	\end{equation}
	where $W_{v}(t_j)$ are the workers assigned to task $t_j$, $p_i$ is the maximum platform revenue which got from the worker completing the task. So the upper bound is $\phi_{\mathbb{Z}} (S^*) =d \cdot \sum_{t_j \in T} p_{<t_j,{W_{v}(t_j)}^*>}=d \cdot \overline p_{t_j}$.
\end{proof}

\textbf{Complexity Analysis.} The time complexity is $O(\vert T \vert \cdot \vert {S_t}\vert \cdot {\vert W \vert}^2+{\vert W\vert}^2 \cdot \vert {S_t}\vert\cdot d \cdot \overline p_{t_j})$, where $\vert T \vert$ is the number of tasks, $\vert W \vert$ is the number of workers, $\vert {S_t} \vert$ is the maximum number of skills, $d \cdot \overline p_{t_j}$ is the number of iterations which adjusts the best response strategy of each worker until the Nash equilibrium is reached.

	\section{Experimental study}
	\label{Sect.6}
	\subsection{Experiment setup}
	We use two datasets in our experiment. For the real dataset, we use the taxi data from Didi Chuxing~\cite{gaiyadidi}, which contains order data in Chengdu from November 1 to November 30, 2016. The order data has information on pick-ups and drop-offs, the start and the end of billing time. We use the pick-up location and the starting time of billing as the location information and arrival time of the task respectively. Since the worker becomes available again after the passenger gets off the car, the drop-off location is used as the location of the worker and the end billing time is used as the arrival time of the worker. For the synthetic dataset, we randomly generate task requests and workers in a rectangular area of Chengdu. The arrival time distribution of workers and task requests follows the uniform distribution in a day. Table~\ref{tab2} depicts our experimental settings, where the default values of parameters are in bold font.
	
	
	\begin{table}
		\caption{Experiments settings}\label{tab2}
		\centering
		\begin{tabular}{|l|l|}
			\hline
			\rule{0pt}{12pt}
			\makecell[c]{\bfseries Parameter} &  \makecell[c]{\bfseries Setting}\\  
			\hline
			\makecell[c]{Fixed range constraint $r$} & \makecell[c]{600, 800, {\bfseries 1000}, 1200, 1400}\\
			\makecell[c]{Extra range constraint $b$} & \makecell[c]{(400,600),600,800),{\bfseries (800,1000)}, (1000,1200), (1200,1400)}\\
			\makecell[c]{Number of skills $S$} & \makecell[c]{10, 11, {\bfseries 12}, 13, 14}\\
			\makecell[c]{Number of tasks $T$} & \makecell[c]{800, 900, {\bfseries 1000}, 1100, 1200}\\
			\makecell[c]{Number of workers $W$} & \makecell[c]{2400, 2700, {\bfseries 3000}, 3300, 3600}\\
			\hline
		\end{tabular}
	\end{table}

	\textbf{Compared algorithms.} We evaluate the performance of the representative algorithms, i.e., Random algorithm (RAN), Greedy algorithm (GRY), and Extra budget-aware Game-Theoretic algorithm (EBGT). All the algorithms are implemented in Java, run on a machine with Intel(R) Core (TM) i7-7700 CPU @ 3.60GHz and 16 GB RAM.
	
	\subsection{Results on the real dataset}

    \subsubsection{Effect of the number of tasks.} In this section, we explore the effect of the number of tasks of the algorithms. As shown in Fig.~\ref{fig6}(a), the running time of algorithms increases with the increment of the number of tasks. EBGT runs longer than RAN and GRY because it needs multiple iterations to find the optimal solution of each round continuously. However, the running time is totally acceptable even when the number of tasks reaches 1200, i.e., about 16 seconds.  In Fig.~\ref{fig6}(b), obviously, the platform revenue increases with the increment of the number of tasks. The reason is that more tasks can be completed which will generate more platform revenue. EBGT gains the highest platform revenue because it finds the best-response task for each worker.
    
    \subsubsection{Effect of the number of workers.} As shown in Fig.~\ref{fig7}(a), the running time of algorithms increases with the increment of the number of workers. EBGT runs longer than RAN and GRY because it needs multi-rounds of iteration to find the optimal match for each worker. From Fig.~\ref{fig7}(b), we can see the platform revenue increases with the increment of the number of workers. Since more workers mean that the platform can select the workers who make the platform more rewarding to complete the tasks. EBGT still gains the most platform revenue because it matches better workers to tasks.
	 \begin{figure}[htbp]
		\centering
		\subfigure[\scriptsize Running time]{
			\begin{minipage}[t]{0.45\linewidth}
				\centering
				\includegraphics[width=5.5cm]{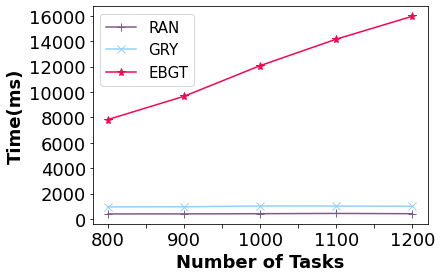}
			\end{minipage}%
		}%
		\subfigure[\scriptsize Total revenue]{
			\begin{minipage}[t]{0.45\linewidth}
				\centering
				\includegraphics[width=5.5cm]{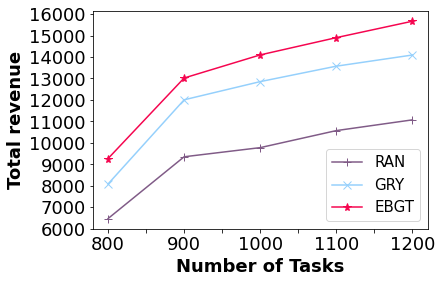}
			\end{minipage}%
		}%
		\setlength{\abovecaptionskip}{0pt}
		\setlength{\belowcaptionskip}{10pt}
		\centering
		\caption{Effect of the number of tasks on real dataset}\label{fig6}
	\end{figure}
		
      \begin{figure}[htbp]
		\centering
		\subfigure[\scriptsize Running time]{
			\begin{minipage}[t]{0.45\linewidth}
				\centering
				\includegraphics[width=5.5cm]{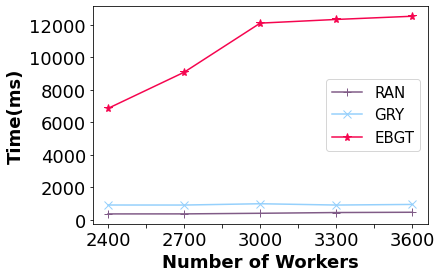}
			\end{minipage}%
		}%
		\subfigure[\scriptsize Total revenue]{
			\begin{minipage}[t]{0.45\linewidth}
				\centering
				\includegraphics[width=5.5cm]{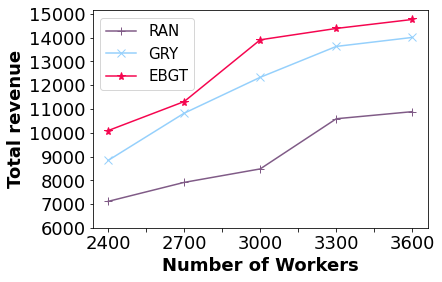}
			\end{minipage}%
		}%
		\setlength{\abovecaptionskip}{0pt}
		\setlength{\belowcaptionskip}{10pt}
		\centering
		\caption{Effect of the number of workers on real dataset}\label{fig7}
	\end{figure}
	
      \begin{figure}[htbp]
		\centering
		\subfigure[\scriptsize Running time]{
			\begin{minipage}[t]{0.45\linewidth}
				\centering
				\includegraphics[width=5.5cm]{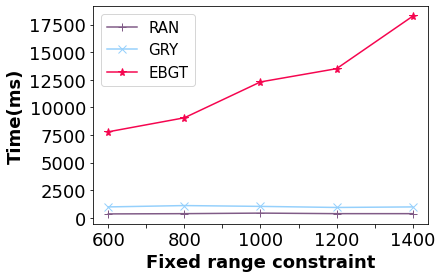}
			\end{minipage}%
		}%
		\subfigure[\scriptsize Total revenue]{
			\begin{minipage}[t]{0.45\linewidth}
				\centering
				\includegraphics[width=5.5cm]{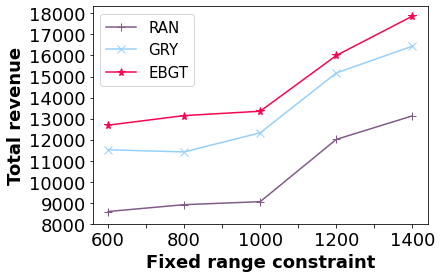}
			\end{minipage}%
		}%
		\setlength{\abovecaptionskip}{0pt}
		\setlength{\belowcaptionskip}{10pt}
		\caption{Effect of the fixed range constraint on real dataset}\label{fig1}
	\end{figure}
	
	\subsubsection{Effect of the fixed range constraint.} As shown in Fig.~\ref{fig1}(a), the running time of algorithms increases with the larger fixed range constraint. This is because more workers will be located in the fixed range constraint of each task which needs longer running time to be processed. EBGT runs slower than RAN and GRY because it needs multiple rounds of iteration to find the optimal solution of each round continuously. As the platform revenue results are shown in Fig.~\ref{fig1}(b), it increases as the fixed range constraint grows. EBGT gains the highest platform revenue for it finds the best-response task for each worker.

	\begin{figure}[htbp]
		\centering
		\subfigure[\scriptsize Running time]{
			\begin{minipage}[t]{0.45\linewidth}
				\centering
				\includegraphics[width=5.5cm]{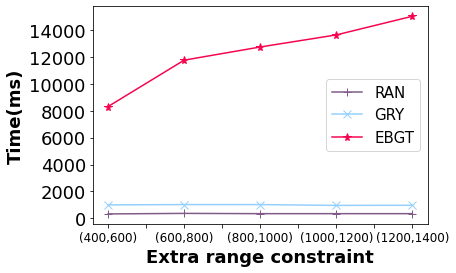}
			\end{minipage}%
		}%
		\subfigure[\scriptsize Total revenue]{
			\begin{minipage}[t]{0.45\linewidth}
				\centering
				\includegraphics[width=5.5cm]{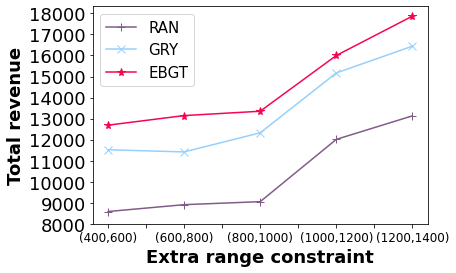}
			\end{minipage}%
		}%
		\setlength{\abovecaptionskip}{0pt}
		\setlength{\belowcaptionskip}{10pt}
		\centering
		\caption{Effect of the extra range constraint on real dataset}\label{fig2}
	\end{figure}
	\begin{figure}[htbp]
		\centering
		\subfigure[Running time]{
			\begin{minipage}[t]{0.45\linewidth}
				\includegraphics[width=\linewidth]{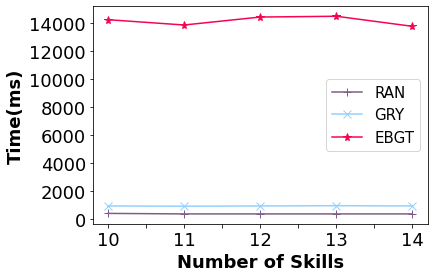}
			\end{minipage}%
		}%
		\subfigure[ Total revenue]{
			\begin{minipage}[t]{0.45\linewidth}
				\includegraphics[width=\linewidth]{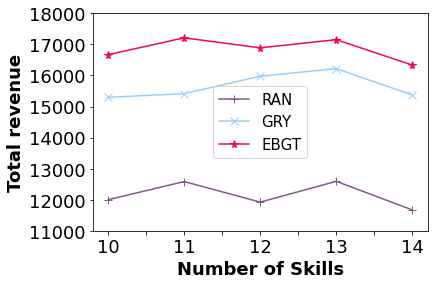}
			\end{minipage}%
		}%
		\setlength{\abovecaptionskip}{0pt}
		\setlength{\belowcaptionskip}{10pt}
		\centering
		\caption{Effect of the number of skills on real dataset}\label{fig3}
	\end{figure}

        \begin{figure}[htbp]
		\centering
		\subfigure[\scriptsize Running time]{
			\begin{minipage}[t]{0.45\linewidth}
				\centering
				\includegraphics[width=5.5cm]{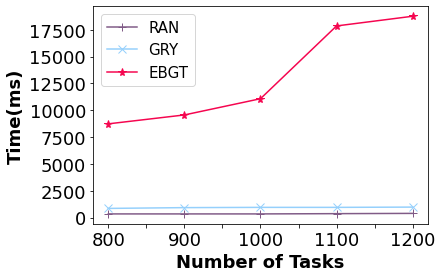}
			\end{minipage}%
		}%
		\subfigure[\scriptsize Total revenue]{
			\begin{minipage}[t]{0.45\linewidth}
				\centering
				\includegraphics[width=5.5cm]{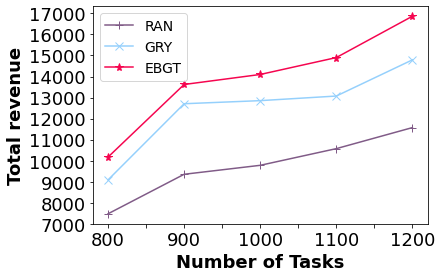}
			\end{minipage}%
		}%
		\setlength{\abovecaptionskip}{0pt}
		\setlength{\belowcaptionskip}{10pt}
		\centering
		\caption{Effect of the number of tasks on synthetic dataset}\label{fig4}
	\end{figure}
 
	\subsubsection{Effect of the extra range constraint.} As depicted in Fig.~\ref{fig2}(a), when the extra range constraint becomes larger, the running time of algorithms increases. This is because more workers will be located in the extra range constraint of each task which needs longer running time. From Fig.~\ref{fig2}(b) we can see that the platform revenue of all the algorithms increases with the larger extra range constraint, this is because more pairs will be matched. EBGT gains keep owning the highest platform revenue as it finds the best-response task for each worker.

	\begin{figure}[htbp]
		\centering
		\subfigure[\scriptsize Running time]{
			\begin{minipage}[t]{0.45\linewidth}
				\centering
				\includegraphics[width=5.5cm]{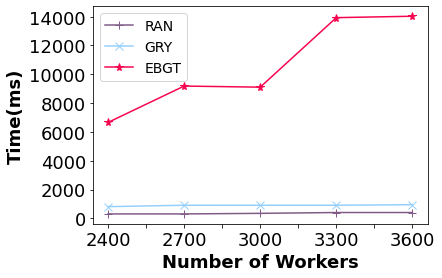}
			\end{minipage}%
		}%
		\subfigure[\scriptsize Total revenue]{
			\begin{minipage}[t]{0.45\linewidth}
				\centering
				\includegraphics[width=5.5cm]{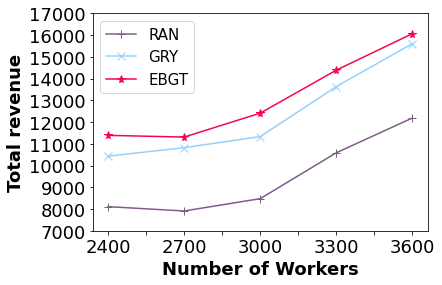}
			\end{minipage}%
		}%
		\setlength{\abovecaptionskip}{0pt}
		\setlength{\belowcaptionskip}{10pt}
		\centering
		\caption{Effect of the number of workers on synthetic dataset}\label{fig5}
	\end{figure}
 
	\subsubsection{Effect of the number of skills.} As we can see in Fig.~\ref{fig3}(a), the running time of algorithms with the number of skills changes little. This is because as the number of skills required for a task increases, the number of tasks that can be satisfied decreases, which leads to a reduction in the number of iterations, which in turn makes the running time stable. EBGT still runs longer than the other two algorithms because it needs multiple rounds of iteration to find the optimal solution of each round continuously. From Fig.~\ref{fig5}(b), we can see that the total revenue of all three algorithms changes a little. EBGT gains the highest platform revenue and RAN gains the least platform revenue.

	\subsection{Results on the synthetic dataset}

	\subsubsection{Effect of the number of tasks.} As we can see in Fig.~\ref{fig4}(a), the running time of algorithms increases with the increment of the number of tasks. Similarly, EBGT runs longer than RAN and GRY. In Fig.~\ref{fig4}(b), the platform revenue increases with the increment of the number of tasks. The reason is that more tasks can be completed which will generate more platform revenue. EBGT still gains the most platform revenue. In summary, the experimental results are similar to those on the real data set above.

	\subsubsection{Effect of the number of workers.} In Fig.~\ref{fig5}(a), the running time of algorithms increases with the increment of the number of workers. EBGT runs longer than RAN and GRY because it needs multiple rounds of iteration to find the optimal solution. From Fig.~\ref{fig5}(b), the platform revenue increases with the increment of the number of workers. The reason is that more workers can complete tasks which will generate more platform revenue. EBGT gains the most platform revenue. In conclusion, the results of the experiment are also similar to the above on the real data set.

\section{Conclusion}
\label{Sect.7}
In this paper, we study the problem of the Skilled Task Assignment with Extra Budget (STAEB) in spatial crowdsourcing, where each task with an extra budget may need multi-skill workers to complete them. We propose two approximation algorithms, including greedy and game-theoretic approaches. Specifically, the greedy approach sorts tasks in order of average fee of skills and greedily assigns fewer workers to cover the skills required by tasks. In addition, we propose a game-theoretic approach to further increase the total platform revenue. Extensive experiments on real and synthetic datasets show that our proposals achieve good efficiency and scalability.
\paragraph*{Supplemental Material Statement:} Source code for STAEB is attached with the submission on ---, our taxi data is available in \cite{gaiyadidi} and our synthetic dataset can be generated with \textbf{gMission} \cite{gMission}.

\bibliography{sn-bibliography}


\end{document}